\documentclass[a4paper]{article}
\usepackage{a4wide}
\usepackage{amsmath,amssymb,amsthm}
\usepackage{thmtools}
\usepackage{enumerate}
\usepackage{abstract}
\usepackage{tikz}
\usepackage{hyperref}
\usepackage[T1]{fontenc}
\usepackage{mathpazo}
\makeatletter\@ifpackageloaded{mathpazo}\@tempswatrue\@tempswafalse
\if@tempswa
  \DeclareFontFamily{OT1}{pzc}{}
  \DeclareFontShape{OT1}{pzc}{m}{it}{<-> s * [1.15] pzcmi7t}{}
  \DeclareMathAlphabet{\mathpzc}{OT1}{pzc}{m}{it}
\fi\makeatother
\usepackage[utf8]{inputenc}

\usepackage[sorting=nyt,firstinits=true,backend=bibtex]{biblatex}
\makeatletter\@ifpackageloaded{biblatex}{%
  \bibliography{references}
  \renewbibmacro{in:}{%
    \ifentrytype{incollection}{\printtext{\bibstring{in}\intitlepunct}}{}}
  \renewbibmacro{publisher+location+date}{%
    \iflistundef{publisher}
      {\setunit*{\addcomma\space}}
      {\setunit*{\addcomma\space}}%
    \printlist{publisher}%
    \setunit*{\addcomma\space}%
    \printlist{location}%
    \setunit*{\addcomma\space}%
    \usebibmacro{date}%
    \newunit}
  \DeclareFieldFormat[article]{pages}{#1\isdot}
  \DeclareFieldFormat[article,incollection,inproceedings,unpublished]{title}{#1\isdot}
  \DeclareFieldFormat[article]{journaltitle}{\mkbibemph{#1\isdot}}
  
}{}\makeatother

\declaretheorem[numberwithin=section,refname={theorem,theorems},Refname={Theorem,Theorems}]{theorem}

\declaretheorem[sibling=theorem,name=Lemma]{lemma}
\declaretheorem[sibling=theorem,name=Proposition]{proposition}
\declaretheorem[sibling=theorem,name=Corollary]{corollary}
\declaretheorem[sibling=theorem,name=Example,style=definition]{example}

\makeatletter\@ifpackageloaded{hyperref}{%
  \usepackage{xcolor}
  \definecolor{dark-red}{rgb}{0.4,0.15,0.15}
  \definecolor{dark-blue}{rgb}{0.15,0.15,0.4}
  \definecolor{medium-blue}{rgb}{0,0,0.5}
  \hypersetup{
    colorlinks,
    linkcolor={dark-red},
    citecolor={dark-blue},
    urlcolor={medium-blue}%
  }

}{}\makeatother

\newcommand{\address}[1]{\vspace{-2em}\begin{center}{\footnotesize #1}\end{center}}

\usepackage{sectsty}
\sectionfont{\large\bfseries}
\subsectionfont{\normalsize\bfseries}

\newcommand{\mirror}[1]{\widetilde{#1}}
\newcommand{\N}{\mathbb{N}}
\newcommand{\Z}{\mathbb{Z}}
\newcommand{\Q}{\mathbb{Q}}
\newcommand{\T}{\mathbb{T}}
\newcommand{\ind}[1]{\text{ind}(w)}
\newcommand{\indq}[1]{\text{ind}_\Q(w)}
\newcommand{\Lang}{\mathcal{L}}

\newcommand{\keywords}[1]{\par\noindent{\footnotesize{\em Keywords\/}: #1}}

\begin{document}
  \title{Characterization of repetitions in Sturmian words:\\A new proof}
  \author{Jarkko Peltomäki\\
          \small \href{mailto:jspelt@utu.fi}{jspelt@utu.fi}}
  \date{}
  \maketitle
  \address{Turku Centre for Computer Science TUCS, 20520 Turku, Finland\\
           University of Turku, Department of Mathematics and Statistics, 20014 Turku, Finland}

  \noindent
  \hrulefill
  \begin{abstract}
    \vspace{-1em}
    \noindent
    We present a new, dynamical way to study powers (that is, repetitions) in Sturmian words based on results from
    Diophantine approximation theory. As a result, we provide an alternative and shorter proof of a result by Damanik
    and Lenz characterizing powers in Sturmian words [Powers in Sturmian sequences, Eur. J. Combin. 24 (2003),
    377--390]. Further, as a consequence, we obtain a previously known formula for the fractional index of a Sturmian
    word based on the continued fraction expansion of its slope.
    \vspace{1em}
    \keywords{sturmian word, standard word, power, combinatorics on words, continued fraction}
    \vspace{-1em}
  \end{abstract}
  \hrulefill

  \let\thefootnote\relax\footnotetext{{\scriptsize \textcopyright\, 2015, Elsevier. Licensed under the
  \href{http://creativecommons.org/licenses/by-nc-nd/4.0/}%
  {Creative Commons Attribution-NonCommercial-NoDerivatives 4.0 International}.}}

  \section{Introduction}
  In 2003 Damanik and Lenz \cite{2003:powers_in_sturmian_sequences} completely described factors of length $n$ of a
  Sturmian word which occur as $p^{th}$ powers for every $n \geq 0$ and $p \geq 1$. Damanik and Lenz prove a series of
  results concerning how factors of a Sturmian word align to the corresponding (finite) standard words. By a careful
  analysis of the alignment, they obtain the complete description of powers thanks to known results on powers of
  standard words. Our method is based on the dynamical view of Sturmian words as codings of irrational rotations.
  Translating word-combinatorial concepts into corresponding dynamical concepts allows us to apply powerful results
  from Diophantine approximation theory (such as the Three Distance Theorem) providing a more geometric proof of the
  result of Damanik and Lenz. Our methods allow us to avoid tricky alignment arguments making the proof in our opinion
  easier to follow. Furthermore, the results allow us to infer a formula for the fractional index of a Sturmian word
  based on the continued fraction expansion of its slope. This formula and its proof appeared in an earlier paper by
  Damanik and Lenz \cite{2002:the_index_of_sturmian_sequences} and was also established purely combinatorially using
  alignment arguments. The formula was independently obtained with different methods by Carpi and de Luca
  \cite{2000:special_factors_periodicity_and_an_application_to_sturmian} and Justin and Pirillo
  \cite{2001:fractional_powers_in_sturmian_words}. For partial results and works related to powers in Sturmian words
  see e.g. the papers of Mignosi \cite{1989:infinite_words_with_linear_subword_complexity}, Berstel
  \cite{1999:on_the_index_of_sturmian_words}, Vandeth \cite{2000:sturmian_words_and_words_with_a_critical_exponent},
  and Justin and Pirillo \cite{2001:fractional_powers_in_sturmian_words}.

  The paper is organized as follows: in \autoref{sec:continued_fractions} we briefly recall results concerning
  continued fractions and rational approximations and prove the purely number-theoretic and important
  \autoref{prp:closest} for later use in \autoref{sec:main_results}. In \autoref{sec:sturmian_words} we state needed
  facts about Sturmian words with appropriate references. \autoref{sec:main_results} contains the main results and
  their proofs.

  \section{Continued Fractions and Rational Approximations}\label{sec:continued_fractions}
  Every irrational real number $\alpha$ has a unique infinite continued fraction expansion
  \begin{align}\label{eq:cf}
    \alpha = [a_0; a_1, a_2, a_3, \ldots] = a_0 + \dfrac{1}{a_1 + \dfrac{1}{a_2 + \dfrac{1}{a_3 + \cdots}}}
  \end{align}
  with $a_0 \in \Z$ and $a_k \in \N$ for all $k \geq 1$. The numbers $a_i$ are called the \emph{partial quotients} of
  $\alpha$. Good references on continued fractions are the books of Khinchin \cite{1997:continued_fractions} and
  Cassels \cite{1957:an_introduction_to_diophantine_approximation}. We focus here only on irrational numbers, but we
  note that with small tweaks much of what follows also holds for rational numbers, which have finite continued fraction
  expansions.
  
  The \emph{convergents} $c_k = \frac{p_k}{q_k}$ of $\alpha$ are defined by the recurrences
  \begin{alignat*}{4}
    p_0 = a_0, &\qquad p_1 = a_1 a_0 + 1, &\qquad p_k = a_k p_{k-1} + p_{k-2}, \qquad& k \geq 2, \\
    q_0 = 1,   &\qquad q_1 = a_1,         &\qquad q_k = a_k q_{k-1} + q_{k-2}, \qquad& k \geq 2.
  \end{alignat*}
  The sequence $(c_k)_{k \geq 0}$ converges to $\alpha$. Moreover, the even convergents are less than $\alpha$ and form
  an increasing sequence and, on the other hand, the odd convergents are greater than $\alpha$ and form a decreasing
  sequence.

  If $k \geq 2$ and $a_k > 1$, then between the convergents $c_{k-2}$ and $c_k$ there are \emph{semiconvergents}
  (called intermediate fractions in Khinchin's book \cite{1997:continued_fractions}) which are of the form
  \begin{align*}
    \frac{p_{k,l}}{q_{k,l}} = \frac{lp_{k-1} + p_{k-2}}{lq_{k-1} + q_{k-2}}
  \end{align*}
  with $1 \leq l < a_k$. When the semiconvergents (if any) between $c_{k-2}$ and $c_k$ are ordered by the size of their
  denominators, the obtained sequence is increasing if $k$ is even and decreasing if $k$ is odd.

  Note that we make a clear distinction between convergents and semiconvergents, i.e., convergents are not a specific
  subtype of semiconvergents.
  
  For the rest of this paper we make the convention that $\alpha$ refers to an irrational number with a continued
  fraction expansion as in \eqref{eq:cf} having convergents $\frac{p_k}{q_k}$ and semiconvergents
  $\frac{p_{k,l}}{q_{k,l}}$ as above.

  A rational number $\frac{a}{b}$ is a \emph{best approximation} of the real number $\alpha$ if for every fraction
  $\frac{c}{d}$ such that $\frac{c}{d} \neq \frac{a}{b}$ and $d \leq b$ it holds that
  \begin{align*}
    \left|b\alpha - a\right| < \left|d\alpha - c\right|.
  \end{align*}
  In other words, any other multiple of $\alpha$ with a coefficient at most $b$ is further away from the nearest
  integer than is $b\alpha$. The next proposition shows that the best approximations of an irrational number are
  connected to its convergents (for a proof see Theorems 16 and 17 of \cite{1997:continued_fractions}).

  \begin{proposition}\label{prp:best_approximation_convergent}
    The best rational approximations of an irrational number are exactly its convergents.
  \end{proposition}

  We identify the unit interval $[0,1)$ with the unit circle $\T$. Let $\alpha \in (0,1)$ be irrational. The map
  \begin{align*}
    R: [0, 1) \to [0, 1), \, x \mapsto \{x + \alpha\},
  \end{align*}
  where $\{x\}$ stands for the fractional part of the number $x$, defines a rotation on $\T$. The circle partitions
  into the intervals $(0,\frac{1}{2})$ and $(\frac{1}{2},1)$. Points in the same interval of the partition are said to
  be on the same side of $0$, and points in different intervals are said to be on the opposite sides of $0$. (We are not
  interested in the location of the point $\frac{1}{2}$.) The points $\{q_k \alpha\}$ and $\{q_{k-1}\alpha\}$ are
  always on the opposite sides of $0$. The points $\{q_{k,l}\alpha\}$ with $0 < l \leq a_k$ always lie between the
  points $\{q_{k-2}\alpha\}$ and $\{q_k \alpha\}$; see \eqref{eq:distance_difference}.

  We measure the shortest distance to $0$ on $\T$ by setting
  \begin{align*}
    \|x\| = \min\{\{x\}, 1 - \{x\}\}. 
  \end{align*}
  We have the following facts for $k \geq 2$ and for all $l$ such that $0 < l \leq a_k$:
  \begin{align}\label{eq:distance_formula}
    \|q_{k,l} \alpha\| &= (-1)^k(q_{k,l} \alpha - p_{k,l}),\\
    \|q_{k,l}\alpha\|  &= \|q_{k,l-1}\alpha\| - \|q_{k-1}\alpha\|. \label{eq:distance_difference}
  \end{align}
  We can now interpret \autoref{prp:best_approximation_convergent} as
  \begin{align}\label{eq:min_distance}
    \min_{0 < n < q_k} \|n\alpha\| = \|q_{k-1}\alpha\|, \quad \text{for } k \geq 1.
  \end{align}
  Note that rotating preserves distances; a fact we will often use without explicit mention. In particular, the
  distance between the points $\{n\alpha\}$ and $\{m\alpha\}$ is $\||n-m|\alpha\|$. Thus by \eqref{eq:min_distance} the
  minimum distance between the distinct points $\{n\alpha\}$ and $\{m\alpha\}$ with $0 \leq n,m < q_k$ is at least
  $\|q_{k-1}\alpha\|$. The formula \eqref{eq:min_distance} tells what is the point closest to $0$ among the points
  $\{n\alpha\}$ for $1 \leq n \leq q_k - 1$. We are also interested to know the point closest to $0$ on the side
  opposite to $\{q_{k-1}\alpha\}$. The next result is very important and concerns this. 

  \begin{proposition}\label{prp:closest}
    Let $\alpha$ be an irrational number. Let $n$ be an integer such that $0 < n < q_{k,l}$ with $k \geq 2$ and
    $0 < l \leq a_k$. If $\|n\alpha\| < \|q_{k,l-1}\alpha\|$, then $n = mq_{k-1}$ for some integer $m$ such that
    $1 \leq m \leq \min\{l, a_k - l + 1\}$.
  \end{proposition}
  \begin{proof}
    Suppose that $\|n\alpha\| < \|q_{k,l-1}\alpha\|$, and assume for a contradiction that the point $\{n\alpha\}$ is on
    the same side of $0$ as $\{q_{k-2}\alpha\}$. Since $n < q_{k,l}$, we conclude that $n \neq q_{k,r}$ for $r \geq l$.
    By \eqref{eq:distance_difference} and our assumption that $\|n\alpha\| < \|q_{k,l-1}\|$, we see that
    $n \neq q_{k,r}$ with $0 \leq r \leq l-1$. As $\|n\alpha\| > \|q_k \alpha\|$ by \eqref{eq:min_distance}, we infer
    that the point $\{n\alpha\}$ must lie between the points $\{q_{k,l'}\alpha\}$ and $\{q_{k,l'+1}\alpha\}$ for some
    $l'$ such that $0 \leq l' < a_k$. The distance between the points $\{n\alpha\}$ and $\{q_{k,l'}\}$ is less than
    $\|q_{k-1}\alpha\|$. By \eqref{eq:min_distance}, it must be that $q_{k,l'} \geq q_k$; a contradiction.

    Suppose for a contradiction that $n$ is not a multiple of $q_{k-1}$. Then the point $\{n\alpha\}$ lies between
    the points $\{tq_{k-1}\alpha\}$ and $\{(t+1)q_{k-1}\alpha\}$ for some $t$ such that
    $0 < t < \lfloor 1/\|q_{k-1}\alpha\| \rfloor$. As $\{n\alpha\}$ is on the same side of $0$ as the point
    $\{q_{k-1}\alpha\}$, it follows that $\|n\alpha\| > \|tq_{k-1}\alpha\|$ and
    $\|tq_{k-1}\alpha\| = t\|q_{k-1}\alpha\|$. The distance between the points $\{n\alpha\}$ and $\{tq_{k-1}\alpha\}$
    is less than $\|q_{k-1}\alpha\|$, so by \eqref{eq:min_distance} it must be that
    $tq_{k-1} \geq q_k = a_kq_{k-1} + q_{k-2}$. Thus necessarily $t > a_k$. Using \eqref{eq:distance_difference} we
    see that the distance between the points $\{q_k\alpha\}$ and $\{q_{k-2}\alpha\}$ is $a_k\|q_{k-1}\alpha\|$. Since
    $\|q_k \alpha\| < \|q_{k-1}\alpha\|$, we infer that
    \begin{align}\label{eq:max_coefficient}
      \|q_{k,l-1}\alpha\| \leq \|q_{k-2}\alpha\| = a_k\|q_{k-1}\alpha\| + \|q_k \alpha\| < (a_k + 1)\|q_{k-1}\alpha\|.
    \end{align}
    Therefore by our assumption,
    \begin{align*}
      (a_k + 1)\|q_{k-1}\alpha\| > \|q_{k,l-1}\alpha\| > \|n\alpha\| > t\|q_{k-1}\alpha\|,
    \end{align*}
    so $a_k \geq t$; a contradiction. We have thus concluded that $n = mq_{k-1}$ for some $m \geq 1$.
    
    Let us now analyze the upper bound on $m$. First of all, $mq_{k-1} < q_{k,l}$ exactly when $m \leq l \leq a_k$. It
    follows that $\|mq_{k-1}\alpha\| = m\|q_{k-1}\alpha\|$. By \eqref{eq:distance_difference}
    \begin{align*}
      m\|q_{k-1}\alpha\| < \|q_{k,l-1}\alpha\| = (a_k - (l - 1))\|q_{k-1}\alpha\| + \|q_k\alpha\|,
    \end{align*}
    so $m \leq a_k - l + 1$. We conclude that $m \leq \min\{l, a_k - l + 1\}$.
  \end{proof}

  The inequalities \eqref{eq:distance_difference} and \eqref{eq:max_coefficient} imply that
  $a_k\|q_{k-1}\alpha\| < \|q_{k-2}\alpha\| < (a_k + 1)\|q_{k-1}\alpha\|$. We derive the following useful fact:
  \begin{align}\label{eq:next_quotient}
    a_k = \left\lfloor \frac{\|q_{k-2}\alpha\|}{\|q_{k-1}\alpha\|} \right\rfloor.
  \end{align}

  We need the famous Three Distance Theorem (see e.g. \cite{1998:three_distance_theorems_and_combinatorics_on_words}  
  and the references therein).

  \begin{theorem}[The Three Distance Theorem]\label{thm:three_distance}
    Let $\alpha$ be an irrational number, and let $n > a_1$ be a positive integer uniquely expressed in the form
    $n = lq_{k-1} + q_{k-2} + r$ with $k \geq 2, 0 < l \leq a_k,$ and $0 \leq r < q_{k-1}$. The points
    $0, \{\alpha\}, \{2\alpha\}, \ldots, \{n\alpha\}$ partition the circle $\T$ into $n+1$ intervals. There are exactly
    \begin{itemize}
      \item $n+1-q_{k-1}$ intervals of length $\|q_{k-1}\alpha\|$,
      \item $r+1$ intervals of length $\|q_{k,l}\alpha\|$, and 
      \item $q_{k-1} - (r+1)$ intervals of length $\|q_{k,l-1}\alpha\|$.
    \end{itemize}
  \end{theorem}

  By \eqref{eq:distance_difference} the intervals of the last type (if they exist) are the longest, and their length is
  the sum of the two other length types.

  \section{Word Combinatorics and Sturmian Words}\label{sec:sturmian_words}
  We mention here only few key concepts from combinatorics on words; good background references are Lothaire's books
  \cite{1983:combinatorics_on_words,2002:algebraic_combinatorics_on_words}.

  A word is \emph{primitive} if it is not a non-trivial power of some word. A word $w$ is primitive if and only if it
  occurs exactly twice in $w^2$. The \emph{cyclic shift operator} $C$ is defined by
  $C(a_1 \cdots a_{n-1} a_n) = a_n a_1 \cdots a_{n-1}$ where the $a_i$ are letters. A word $u$ is \emph{conjugate} to
  $v$ if $C^i(v) = u$ for some $i$ such that $0 \leq i < |v|$. The \emph{reversal} of the word
  $w = a_1 \cdots a_{n-1} a_n$ where the $a_i$ are letters is the word $\mirror{w} = a_n a_{n-1} \cdots a_1$.

  Sturmian words are a well-known class of infinite, aperiodic binary words over $\{0,1\}$ with minimal factor
  complexity. They are defined as the (right-)infinite words having $n+1$ factors of length $n$ for every $n \geq 0$.
  For our purposes it is more convenient to view Sturmian words equivalently as the infinite words obtained as codings
  of orbits of points in an irrational circle rotation with two intervals
  \cite{2002:substitutions_in_dynamics_arithmetics_and_combinatorics,2002:algebraic_combinatorics_on_words}.
  Let us make this more precise. The frequency $\alpha$ of letter $1$ (called the \emph{slope}) in a Sturmian words
  exists, and it is irrational. Divide the circle $\T$ into two intervals $I_0$ and $I_1$ defined by the points $0$ and
  $1-\alpha$, and define the coding function $\nu$ by setting $\nu(x) = 0$ if $x \in I_0$ and $\nu(x) = 1$ if
  $x \in I_1$. The coding of the orbit of a point $x$ is the infinite word $s_{x,\alpha}$ obtained by setting its
  $n^\text{th}, n \geq 0,$ letter to equal $\nu(R^n(x))$ where $R$ is the rotation by angle $\alpha$. This word is
  Sturmian with slope $\alpha$, and conversely every Sturmian word with slope $\alpha$ is obtained this way. To make the
  definition proper, we need to define how $\nu$ behaves in the endpoints $0$ and $1-\alpha$. We have two options:
  either take $I_0 = [0,1-\alpha)$ and $I_1 = [1-\alpha,1)$ or $I_0 = (0,1-\alpha]$ and $I_1 = (1-\alpha,1]$. The
  difference is seen in the codings of the orbits of the special points $\{-n\alpha\}$, and both options are needed to
  be able to obtain every Sturmian word of slope $\alpha$ as a coding of a rotation. However, in this paper we are not
  concerned about this choice. We make the convention that $I(x, y)$ with $x, y \neq 0$ is either of the half-open
  intervals of $\T$ separated by the points $x$ and $y$ (taken modulo $1$ if necessary) not containing the point $0$ as
  an interior point. The interval $I(x,0) = I(0,x)$ is either of the half-open intervals separated by the points $0$
  and $x$ having smallest length (the case $x = \frac{1}{2}$ is not important in this paper). Since the sequence
  $(\{n\alpha\})_{n\geq 0}$ is dense in $[0,1)$---as is well-known---every Sturmian word of slope $\alpha$ has the same
  language (that is, the set of factors); this language is denoted by $\Lang(\alpha)$. Thus to study repetitions, it
  is sufficient to analyze $\Lang(\alpha)$. The \emph{fractional index} $\indq{w}$ of a nonempty factor
  $w \in \Lang(\alpha)$ is defined as
  \begin{align*}
    \indq{w} = \sup\{k \in \Q\colon\, w^k \in \Lang(\alpha)\}
  \end{align*}
  where the \emph{fractional power} $w^k$ is the word $(uv)^n u$ with $w=uv$ and $k = n + \frac{|u|}{|w|}$. The
  \emph{index} $\ind{w}$ of a nonempty factor $w$ is defined similarly by letting $k$ take only integral values. The
  index of a factor in $\Lang(\alpha)$ is always finite. The \emph{fractional index of a Sturmian word} with slope
  $\alpha$ is defined to be
  \begin{align*}
    \sup\{\indq{w}\colon\, w \in \Lang(\alpha)\}.
  \end{align*}
  This quantity can be infinite.

  For every factor $w = a_0 a_1 \cdots a_{n-1}$ of length $n$ there exists a unique subinterval $[w]$ of $\T$ such that
  $s_{x,\alpha}$ begins with $w$ if and only if $x \in [w]$. Clearly
  \begin{align*}
    [w] = I_{a_0} \cap R^{-1}(I_{a_1}) \cap \ldots \cap R^{-(n-1)}(I_{a_{n-1}}).
  \end{align*}
  We denote the length of the interval $[w]$ by $|[w]|$. The points
  $0, \{-\alpha\}, \{-2\alpha\}, \ldots, \{-n\alpha\}$ partition the circle into $n+1$ intervals, which have one-to-one
  correspondence with the words of $\Lang(\alpha)$ of length $n$. Among these intervals the interval containing the
  point $\{-(n+1)\alpha\}$ corresponds to the right special factor of length $n$. A factor $w$ is \emph{right special}
  if both $w0, w1 \in \Lang(\alpha)$. Similarly a factor is \emph{left special} if both $0w, 1w \in \Lang(\alpha)$. In
  a Sturmian word there exists a unique right special and a unique left special factor of length $n$ for all
  $n \geq 0$. The language $\Lang(\alpha)$ is mirror-invariant, that is, for every $w \in \Lang(\alpha)$ also
  $\mirror{w} \in \Lang(\alpha)$. It follows that the right special factor of length $n$ is the reversal of the left
  special factor of length $n$.

  Given the continued fraction expansion of an irrational $\alpha \in (0,1)$, we define the corresponding
  \emph{standard sequence} $(s_k)_{k\geq 0}$ of words by
  \begin{alignat*}{5}
    s_{-1} = 1, \qquad& s_0 = 0, \qquad& s_1 = s_0^{a_1-1}s_{-1}, \qquad& s_k = s_{k-1}^{a_k}s_{k-2}, \qquad&k \geq 2.
  \end{alignat*}
  As $s_k$ is a prefix of $s_{k+1}$ for $k \geq 1$, the sequence $(s_k)$ converges to a unique infinite word $c_\alpha$
  called the infinite standard Sturmian word of slope $\alpha$, and it equals $s_{\alpha,\alpha}$. Inspired by the
  notion of semiconvergents, we define \emph{semistandard} words for $k \geq 2$ by
  \begin{align*}
    s_{k,l} = s_{k-1}^l s_{k-2}
  \end{align*}
  with $1 \leq l < a_k$. Clearly $|s_k| = q_k$ and $|s_{k,l}| = q_{k,l}$. Every prefix of $c_\alpha$ is left special,
  so in particular, standard and semistandard words are left special. Every standard or semistandard word is primitive
  \cite[Proposition~2.2.3]{2002:algebraic_combinatorics_on_words}. An important property of standard words is that the
  words $s_k$ and $s_{k-1}$ almost commute; namely $s_k s_{k-1} = wab$ and $s_{k-1} s_k = wba$ for some word $w$ and
  distinct letters $a$ and $b$. For more information about standard words see Chapter 2 of
  \cite{2002:algebraic_combinatorics_on_words} and Berstel's paper \cite{1999:on_the_index_of_sturmian_words}. Here we
  see that the only difference between the words $c_\alpha$ and $c_{\overline{\alpha}}$ where
  $\alpha = [0; 1, a_2, a_3, \ldots]$ and $\overline{\alpha} = [0; a_2 + 1, a_3, \ldots]$ is that the roles of the
  letters $0$ and $1$ are reversed. Thus for the study of powers, we may assume without loss of generality that
  $a_1 \geq 2$.

  For the rest of this paper we make the convention that the partial quotients of an irrational $\alpha$ satisfy
  $a_0 = 0$ and $a_1 \geq 2$, that is, $0 < \alpha < \frac{1}{2}$. Moreover, the words $s_k$ and $s_{k,l}$ refer to the
  standard or semistandard words of slope $\alpha$.

  \section{The Main Results}\label{sec:main_results}
  This section presents a complete description of powers occurring in a Sturmian word with slope $\alpha$. As a
  side-product, in \autoref{thm:conjugate_interval_lengths} we obtain a description of conjugacy classes of length
  $q_{k,l}$. Finally, as an easy consequence of the established results, we obtain a formula for the fractional index
  of a Sturmian word (\autoref{thm:index_formula}).

  The following important proposition shows the usefulness of \autoref{prp:closest} in the study of Sturmian words and
  plays a role similar to Theorem~1 of \cite{2003:powers_in_sturmian_sequences}.

  \begin{proposition}\label{prp:square_length}
    If $w^2 \in \Lang(\alpha)$ with $w$ primitive, then $|w| = q_k$ for some $k \geq 0$ or $|w| = q_{k,l}$ for some
    $k \geq 2$ with $0 < l < a_k$.
  \end{proposition}
  \begin{proof}
    Let $n = |w|$. If $n < q_1 = a_1$, then the factors of length $n$ are readily seen to be $0^n$ and the conjugates
    of $0^{n-1}1$. Since the minimum number of letters $0$ between two occurrences of letter $1$ in $\Lang(\alpha)$ is
    $a_1 - 1$ and the maximum number is $a_1$, the only way $w^2$ can be a factor is that $w = 0 = s_0$. Suppose then
    that $n \geq q_1$ and $[w] = I(-i\alpha, -j\alpha)$ with $0 \leq i,j \leq n$. We may assume without loss of
    generality that $w$ is right special, so $\{-(n+1)\alpha\} \in [w]$. Further, since
    $[w^2] = [w] \cap R^{-n}([w]) \neq \emptyset$, then necessarily (depending on $n$) either
    $[w^2] = I(-i\alpha, -(j+n)\alpha)$ or $[w^2] = I(-j\alpha, -(i+n)\alpha)$. We assume that
    $[w^2] = I(-i\alpha, -(j+n)\alpha)$; the other case is symmetric. We wish to prove that the points $\{-(n+1)\alpha$
    and $\{-(j+n)\alpha\}$ are actually the same point. This is equivalent to saying that $j = 1$. Assume on the
    contrary that $j \neq 1$. Let $a$ be the first letter of $w$ and $b$ be a letter such that $b \neq a$. Note that
    $[w^2] \subset [wa]$. Now as $w$ is right special, $[wa] = I(-j\alpha, -(n+1)\alpha)$ and
    $[wb] = I(-(n+1)\alpha, -i\alpha)$. Let $x \in [w^2]$ and $y \in [wa]\setminus[w^2]$, and let $u$ be the longest
    common prefix of $s_{x,\alpha}$ and $s_{y,\alpha}$. Since $[w^2] \neq [wa]$, we have that $|u| < 2|w|$. Moreover,
    $u$ is right special, so $w$ is a suffix of $u$. However, $w^2$ is a prefix of $s_{x,\alpha}$ implying that $u$ is
    a prefix of $w^2$. Thus $w^2$ contains at least three occurrences of $w$ contradicting the primitivity of $w$. From
    this contradiction we conclude that $j = 1$. There are no points $\{-m\alpha\}$ in the interval
    $I(-(j+n)\alpha, -j\alpha) = I(-(n+1)\alpha, -\alpha)$ with $m \leq n$. Therefore the point $\{-n\alpha\}$ is the
    closest point to $0$ from either side. If $q_1 \leq n < q_{2,1}$, then it must be that $n = q_1$. Otherwise let
    $k \geq 2$ be such that $q_{k,l} \leq n < q_{k,l+1}$ with $0 < l \leq a_k$. By \autoref{prp:closest} either
    $n = q_{k-1}$ or $n = q_{k,l}$ proving the claim.
  \end{proof}

  Indeed, for each length given in the statement of the previous proposition, there exists a factor occurring as a
  square.

  \begin{lemma}\label{lem:standard_square}
    We have that $s_0^2, s_1^2 \in \Lang(\alpha)$ and $s_{k,l}^2 \in \Lang(\alpha)$ for all $k \geq 2$ and $l$ such
    that $0 < l \leq a_k$.
  \end{lemma}
  \begin{proof}
    As $s_0^2 = 0^2$ and $s_1^2 = (0^{a_1 - 1}1)^2$, clearly $s_0^2, s_1^2 \in \Lang(\alpha)$. Since the words
    $s_{k+1} s_k$ and $s_k s_{k+1}$ differ only by their last two letters, it follows that $s_k^2$ is a prefix of
    $s_{k+1} s_k$ if $k \geq 2$. As $s_k$ is a prefix of $s_{k+1}$ when $k \geq 0$, the word
    $s_{k,l} = s_{k-1}^l s_{k-2}$ is both a prefix and a suffix of $s_k = s_{k-1}^{a_k} s_{k-2}$ for all $k \geq 2$ and
    $l$ such that $0 < l \leq a_k$. Thus $s_k^2$ contains $s_{k,l}^2$. The claim follows.
  \end{proof}

  As was seen in the proof of \autoref{prp:square_length}, the index of a factor $w$ of length $n$ depends only on the
  maximum $r \geq 0$ such that $R^{-tn}(x) \in [w]$ for $0 \leq t \leq r$ where $x$ is either of the endpoints of
  $[w]$. That is, the index of a factor depends only on the length of its interval but not on its position. To put it
  more precisely, if $w$ is a factor of length $n$, then
  \begin{align}\label{eq:index}
    \ind{w} = \gamma + \left\lfloor \frac{|[w]|}{\||w|\alpha\|} \right\rfloor
  \end{align}
  where $\gamma$ is $1$ if $|[w]| \neq \||w|\alpha\|$ and $0$ otherwise. Next we will carefully characterize the
  lengths of the intervals of factors of length $q_{k,l}$. After this it is easy to conclude the main results.

  \begin{theorem}\label{thm:conjugate_interval_lengths}
    Let $n = q_{k,l}$ with $k \geq 2$ and $0 < l \leq a_k$. Then $C^i(\mirror{s}_{k,l}) \in \Lang(\alpha)$ for
    $0 \leq i \leq n - 1$. The intervals of the first $q_{k-1} - 1$ conjugates of $\mirror{s}_{k,l}$ have length
    $\|q_{k,l-1}\alpha\|$, and the intervals of the latter $n + 1 - q_{k-1}$ conjugates have length
    $\|q_{k-1}\alpha\|$. The interval of the remaining factor has length $\|q_{k,l}\alpha\|$.
  \end{theorem}
  \begin{proof}
    The geometric ideas of this proof are illustrated in the example following this proof. The intervals of the factors
    of length $n$ are called in this proof \emph{level $n$ intervals}. With the same effort we prove here more than
    what is claimed above; we give the exact positions of the intervals of the conjugates of $\mirror{s}_{k,l}$ on the
    circle.

    By \autoref{prp:closest} the interval $J = I(-q_{k,l-1}\alpha, 0)$ has exactly one point $\{-t\alpha\}$ with
    $0 < t \leq n$ as an interior point; namely the point $\{-n\alpha\}$. That is, the point $\{-n\alpha\}$ split the
    level $n-1$ interval $J$ into the level $n$ intervals $K = I(-q_{k,l-1}\alpha, -n\alpha)$ and $L = I(-n\alpha, 0)$.
    Observe that $\|q_{k,l-1}\alpha\| = |J| = |K| + |L| = \|q_{k-1}\alpha\| + \|q_{k,l}\alpha\|$. The Three Distance
    Theorem tells that level $n$ intervals have lengths $\|q_{k,l-1}\alpha\|, \|q_{k,l}\alpha\|,$ and
    $\|q_{k-1}\alpha\|$. In particular, interval $L$ is the unique level $n$ interval of length $\|q_{k,l}\alpha\|$.
    Let $i$ be the smallest positive integer such that the interval $R^{-i}(J)$ is not any interval of level $n$. The
    interval $R^{-i}(J)$ must be a union of two level $n$ intervals: one having length $|K|$ and the other having
    length $|L|$. This is true as by \eqref{eq:distance_difference} it can be deduced that $|J|$ is never a multiple of
    $|K|$; further, $|K|$ is never a multiple of $|L|$. Since the interval of length $\|q_{k,l}\alpha\|$ is unique, we
    conclude that the other interval in the union is $L$. As $\alpha$ is irrational, $R^{-i}(J) \neq J$, so it must be
    that $R^{-i}(J) = M \cup L$ where $M = I(-q_{k-1}\alpha, 0)$. Therefore $R^{-i}$ maps the endpoint $0$ of $L$ to
    the endpoint $\{-q_{k-1}\alpha\}$ of $M$, so $i = q_{k-1}$. As $k > 1$, also $i > 1$. We have shown that the level
    $n$ intervals $R^{-1}(J), R^{-2}(J), \ldots, R^{-(i-1)}(J)$ have length $|J| = \|q_{k-1,l}\alpha\|$. By the Three
    Distance Theorem the remaining $n+1-q_{k-1}$ intervals excluding $L$ have length $\|q_{k-1}\alpha\|$.

    What remains is to analyze the connection between rotation and conjugation. Let $u$ and $v$ be factors of length
    $n$ such that $[u] = M$ and $[v] = L$. Since the intervals $M$ and $L$ are on the opposite sides of $0$, we have
    that $u = au'$ and $v = bv'$ for distinct letters $a$ and $b$. Let $x \in M$ and $y \in L$. Since $i > 1$, the
    interval $R^{-(i-1)}(J) = R(M \cup L)$ is the interval of some factor $w$ of length $n$. Therefore the Sturmian
    words $s_{x+\alpha,\alpha}$ and $s_{y+\alpha,\alpha}$ both have $w$ as a prefix. Thus $s_{x,\alpha}$ begins with
    $aw$ and $s_{y,\alpha}$ begins with $bw$. Hence $w$ must be left special, that is, $w = s_{k,l}$. We will show next
    that $v$ is not conjugate to $s_{k,l}$. Note that $k-1$ is odd if and only if $\{-q_{k-1}\alpha\} \in I_0$. Hence
    the first letter of $v$ is $0$ if and only if $k-1$ is odd. On the other hand, the last letter of $s_{k,l}$ is $0$
    if and only if $k-1$ is even. Thus we conclude that the first letter of $v$ is distinct from the last letter of
    $s_{k,l}$. However, as the suffix of $v$ of length $n-1$ is a prefix of $s_{k,l}$, we see that there are more
    letters $b$ in $v$ than there are in $s_{k,l}$, so $v$ and $s_{k,l}$ cannot be conjugate.
 
    Let then $z$ be the factor of length $n$ such that $[z] = R^{-1}(J)$. Since $\{-n\alpha\} \in J$, it must be that
    $\{-(n+1)\alpha\} \in R^{-1}(J) = [z]$. Thus $z$ is right special, that is, $z = \mirror{s}_{k,l}$. By
    \autoref{lem:standard_square} $s_{k,l}^2 \in \Lang(\alpha)$. Thus every conjugate of $s_{k,l}$ is a factor.
    Further, by the mirror-invariance of $\Lang(\alpha)$, we see that $s_{k,l}$ and $\mirror{s}_{k,l}$ are conjugates.
    Moreover, every conjugate of $\mirror{s}_{k,l}$ is extended to the left by its last letter.

    Suppose that $\lambda \neq v$ is a factor of length $n$ such that $R^{-1}([\lambda])$ is the interval of some
    factor $\mu$ of length $n$. As $R^{-1}([s_{k,l}])$ does not satisfy this condition, it follows that
    $\lambda \neq s_{k,l}$, so $\lambda$ extends to the left uniquely. We will prove that $C(\lambda) = \mu$. Write
    $\lambda = \lambda' c$ for some letter $c$. Then obviously $\mu = d\lambda'$ for some letter $d$. By definition
    $\mu$ must be followed by the letter $c$, that is, $\mu c = d\lambda' c = d \lambda \in \Lang(\alpha)$. We have
    that $d = c$ because $\lambda$ is uniquely extended to the left by its last letter. Therefore we conclude that
    $C(\lambda) = \mu$. In this way we see that the factors of length $n$ having the intervals
    $R^{-1}(J), R^{-2}(J), \ldots, R^{-(i-1)}(J)$ correspond (in order) to the factors
    $\mirror{s}_{k,l}, C(\mirror{s}_{k,l}), \ldots, C^{q_{k-1} - 2}(\mirror{s}_{k,l}) = s_{k,l}$. We saw above that $v$
    is not conjugate to $s_{k,l}$, so it must be that $C(s_{k,l}) = u$. Thus the factors of length $n$ having the
    intervals $[u] = L, R^{-1}(L), R^{-2}(L), \ldots, R^{-(n-q_{k-1})}(L)$ correspond (in order) to the factors
    $u, C(u), C^2(u), \ldots, C^{n-q_{k-1}}(u)$. As $u = C(s_{k,l}) = C^{q_{k-1} - 1}(\mirror{s}_{k,l})$, we have a
    complete description of the positions of the intervals of conjugates of $\mirror{s}_{k,l}$ using the backward orbit
    of $J$ under $R$.
  \end{proof}

  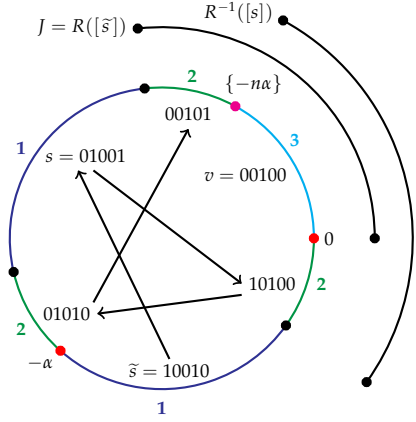
\begin{figure}
  \centering
  \begin{tikzpicture}
    \definecolor{pink}{RGB}{236,0,140}
    \definecolor{bluish}{RGB}{46,49,146}
    \definecolor{cyanish}{RGB}{0,174,239}
    \definecolor{greenish}{RGB}{0,148,69}
    \tikzstyle{own}=[line width=0.8pt]
    \tikzstyle{own2}=[line width=2.0pt,black]
    \tikzstyle{type1}=[line width=0.8pt,bluish,text=bluish,font=\bf]
    \tikzstyle{type2}=[line width=0.8pt,greenish,text=greenish,font=\bf]
    \tikzstyle{type3}=[line width=0.8pt,cyanish,text=cyanish,font=\bf]
    \draw[type3] (0.0:2.0) arc (0.0:61.1:2.0);
    \draw[type2] (61.1:2.0) arc (61.1:96.5:2.0);
    \draw[type1] (96.5:2.0) arc (96.5:192.9:2.0);
    \draw[type2] (192.9:2.0) arc (192.9:228.2:2.0);
    \draw[type1] (228.2:2.0) arc (228.2:324.7:2.0);
    \draw[type2] (324.7:2.0) arc (324.7:360:2.0);
    \filldraw[own2,red] (0:2.0) circle(0.9pt);
    \filldraw[own2,pink] (61.1:2.0) circle(0.9pt);
    \filldraw[own2] (96.5:2.0) circle(0.9pt);
    \filldraw[own2,red] (228.2:2.0) circle(0.9pt);
    \filldraw[own2] (192.9:2.0) circle(0.9pt);
    \filldraw[own2] (324.7:2.0) circle(0.9pt);
    \draw[own] (0:2.8) arc (0:96.5:2.8);
    \draw[own] (-35.3:3.3) arc (-35.3:61.1:3.3);
    \filldraw[own2] (0:2.8) circle(0.9pt);
    \filldraw[own2] (96.5:2.8) circle(0.9pt);
    \filldraw[own2] (-35.3:3.3) circle(0.9pt);
    \filldraw[own2] (61.1:3.3) circle(0.9pt);
    \node at (2.2,0) {\scriptsize{$0$}};
    \node at (-1.6,-1.65) {\scriptsize{$-\alpha$}};
    \node at (1.2,2.05) {\scriptsize{$\{-n\alpha\}$}};
    \node at (1.1,0.85) {\scriptsize{$v = 00100$}};
    \node at (0.35,1.65) {\scriptsize{$00101$}};
    \node at (-1.0,1.1) {\scriptsize{$s = 01001$}};
    \node at (-1.25,-1.0) {\scriptsize{$01010$}};
    \node at (0.1,-1.75) {\scriptsize{$\mirror{s} = 10010$}};
    \node at (1.45,-0.6) {\scriptsize{$10100$}};
    \draw[own] [->] (0.1,-1.6) -- (-1.1,0.9);
    \draw[own] [->] (-0.9,0.9) -- (1.05,-0.6);
    \draw[own] [->] (1.05,-0.75) -- (-0.85,-1.0);
    \draw[own] [->] (-0.9,-0.85) -- (0.35,1.45);
    \node at (-1.05,2.8) {\scriptsize{$J = R([\,\mirror{s}\,])$}};
    \node at (1.0,3.0) {\scriptsize{$R^{-1}([s])$}};
    \node[type1] at (-1.85,1.2) {\scriptsize{1}};
    \node[type1] at (0.0,-2.25) {\scriptsize{1}};
    \node[type2] at (0.4,2.15) {\scriptsize{2}};
    \node[type2] at (-1.85,-1.2) {\scriptsize{2}};
    \node[type2] at (2.1,-0.6) {\scriptsize{2}};
    \node[type3] at (1.75,1.3) {\scriptsize{3}};
  \end{tikzpicture}
  \caption{An example of the geometric ideas in the proof of \autoref{thm:conjugate_interval_lengths}.}
  \label{fig:intervals}
  \end{figure}

  \begin{example}
    Let $\alpha = [0;\overline{2,1}] = \frac{1}{2}(\sqrt{3} - 1)$ (that is, the continued fraction expansion of
    $\alpha$ has period $2,1$). Consider the semiconvergent $\frac{p_{3,1}}{q_{3,1}} = \frac{1+1}{3+2} = \frac{2}{5}$
    of $\alpha$ and factors of length $n = 5$. The factors of length $n$ are $00100, 00101, 01001, 01010, 10010,$ and
    $10100$. Their intervals are depicted in \autoref{fig:intervals}. There are intervals of type $1, 2,$ and $3$
    depending on their length. Intervals of type $1$ have length $\|2\alpha\|$, intervals of type $2$ have length
    $\|3\alpha\|,$ and the interval of type $3$ has length $\|2\alpha\| - \|3\alpha\| = \|5\alpha\|$. As in the proof
    of \autoref{thm:conjugate_interval_lengths}, the point $\{-n\alpha\}$ has split the type $1$ interval
    $J = I(0,-2\alpha)$ into intervals of type $2$ and $3$. The interval $R^{-1}(J)$ corresponds to the right special
    factor $\mirror{s} = \mirror{s}_{3,1}$. The arrows in the figure indicate how conjugation acts on $\mirror{s}$. The
    backward orbit of $J$ corresponds to conjugates of $\mirror{s}$ of type $1$ until the interval of the left special
    factor $s$ is encountered. As seen in the proof of \autoref{thm:conjugate_interval_lengths}, the interval
    $R^{-1}([s])$ no longer coincides with any interval of length $n$. Here
    $R^{-1}([s]) = L \cup M = I(0, -n\alpha) \cup I(0, -3\alpha)$ just as the proof requires. The factor having
    interval $L$ is here seen to be not conjugate to $\mirror{s}$ as it should be by the proof. The interval $M$ must
    then correspond to the conjugate of $s$. As in the proof, the rest of the conjugates of $\mirror{s}$ are obtained
    by rotating $M$ backwards. The intervals obtained this way are of type $2$.
  \end{example}

  We are now ready to prove the main result. The result was originally proven by
  \citeauthor{2003:powers_in_sturmian_sequences} \cite{2003:powers_in_sturmian_sequences}. We present it here phrased
  in a different way.

  \begin{theorem}\label{thm:power_characterization}
    Consider indices of factors of length $n > 0$ in $\Lang(\alpha)$, and let $k \geq 2$.
    \begin{enumerate}[(i)]
      \item If $n < q_1$, then the index of the conjugates of $0^{n-1}1$ is $1$, and the index of the remaining factor
            $0^n$ is $\lfloor a_1 / n \rfloor$.
      \item If $n = q_1$, then the index of the conjugates of $\mirror{s}_1$ is $a_2 + 1$, and the index of the
            remaining factor $0^{a_1}$ is $1$.
      \item If $n = q_k$, then the index of any of the first $q_{k-1} - 1$ conjugates of $\mirror{s}_k$ is
            $a_{k+1} + 2$, the index of any of the remaining $n + 1 - q_{k-1}$ conjugates is $a_{k+1} + 1,$ and
            the index of the remaining factor is $1$.
      \item If $n = q_{k,l}$ with $0 < l < a_k$, then the index of the first $q_{k-1} - 1$ conjugates of
            $\mirror{s}_{k,l}$ is $2$, and the index of the remaining factors is $1$.
      \item If $n = mq_1$ with $1 < m < a_2 + 1$, then the index of any of the first $q_1$ conjugates of
            $\mirror{s}_1^{\,m}$ is $\lfloor(a_2~+~1)/m\rfloor$, and the index of any remaining factor is $1$.
      \item If $n = mq_k$ with $1 < m < a_{k+1} + 2$, then the index of any of the first $q_{k-1} - 1$ conjugates of
            $\mirror{s}_k^{\,m}$ is $\lfloor (a_{k+1} + 2)/m \rfloor$, the index of any of the next $q_k + 1 - q_{k-1}$
            conjugates is $\lfloor(a_{k+1}~+~1)/m\rfloor$, and the index of any remaining factor is $1$.
      \item If $n$ does not fall into any of the above cases, then the index of every factor of length $n$ is $1$.
    \end{enumerate}
  \end{theorem}
  \begin{proof}
    First of all, observe that all the cases \emph{(i)--(vii)} are mutually exclusive. Consider the cases \emph{(i)}
    and \emph{(ii)}. The factors of length $n \leq q_1$ are readily seen to be $0^n$ and the conjugates of $0^{n-1}1$.
    As the index of $0$ is $a_1$, the index of the factor $0^n$ is $\lfloor a_1 / n \rfloor$. The intervals of the
    conjugates of $0^{n-1}1$ have length $\alpha$. If $n = 1$, then the index of the factor $0^{n-1}1 = 1$ is $1$. If
    $n > 1$, the number $1 + \lfloor \alpha / \|n\alpha\| \rfloor$ equals $1$ unless $n = q_1$ when it equals $a_2 + 1$
    by \eqref{eq:next_quotient}. The claims of \emph{(i)} and \emph{(ii)} follow from \eqref{eq:index}.
    
    Consider next a factor $w$ of length $n = q_{k,l}$ for some $k \geq 2$ and $l$ such that $0 < l \leq a_k$. By
    \autoref{thm:conjugate_interval_lengths}, the intervals of the first $q_{k-1} - 1$ conjugates of $\mirror{s}_{k,l}$
    have length $\|q_{k,l-1}\alpha\|$. Using
    \eqref{eq:index} and \eqref{eq:distance_difference} we see that their index equals to
    \begin{align*}
      1 + \left\lfloor \frac{\|q_{k,l-1}\alpha\|}{\|q_{k,l}\alpha \|} \right\rfloor
      = 1 + \left\lfloor \frac{\|q_{k,l}\alpha\| + \|q_{k-1}\alpha\|}{\|q_{k,l} \alpha\|} \right\rfloor
      = 2 + \left\lfloor \frac{\|q_{k-1}\alpha\|}{\|q_{k,l} \alpha\|} \right\rfloor.
    \end{align*}
    If $l \neq a_k$, then by \eqref{eq:min_distance} $\|q_{k,l}\alpha\| > \|q_{k-1}\alpha\|$, so the index is $2$. If
    $l = a_k$, then by \eqref{eq:next_quotient}, the index equals to $2 + a_{k+1}$. This proves the first claims in
    \emph{(iii)} and \emph{(iv)}. The latter cases are analogous, so \emph{(iii)} and \emph{(iv)} are proved.

    \autoref{prp:square_length} shows that the factors not covered by the cases \emph{(i)--(iv)} having index higher
    than $1$ must be nonprimitive. By \emph{(i)} and \emph{(iv)} they must have length $mq_k$ for some $k \geq 1$,
    meaning that we are in either of the cases \emph{(v)} or \emph{(vi)}. It is a straightforward application of
    \emph{(ii)} and \emph{(iii)} to deduce \emph{(v)} and \emph{(vi)}. The theorem is proved.
  \end{proof}

  In particular, every Sturmian word contains infinitely many cubes, but fourth powers are avoidable. The theorem
  implies the following weaker version which is still useful (compare to
  \cite[Lemma~3.6]{2002:the_index_of_sturmian_sequences}):

  \begin{corollary}
    Let $w \in \Lang(\alpha)$ be primitive. If $w^2 \in \Lang(\alpha)$, then $w$ is conjugate to $s_k$ for some
    $k \geq 0$ or to $s_{k,l}$ with $k \geq 2$ and $0 < l < a_k$. If $w^3 \in \Lang(\alpha)$, then either $w = 0$ and
    $a_1 > 2$ or $w$ is conjugate to some $s_k$ with $k \geq 1$. \qed
  \end{corollary}

  We obtain the result of \cite{2002:the_index_of_sturmian_sequences},
  \cite{2000:special_factors_periodicity_and_an_application_to_sturmian}, and
  \cite{2001:fractional_powers_in_sturmian_words} on the fractional index of Sturmian words as a direct consequence of
  the results so far:

  \begin{theorem}\label{thm:index_formula}
    The fractional index of a Sturmian word with slope $\alpha$ is
    \begin{align*}
      \max\left\{a_1, 2 + \sup_{k \geq 2}\{a_k + (q_{k-1} - 2)/q_k\}\right\}.
    \end{align*}
  \end{theorem}
  \begin{proof}
    The largest fractional power of a factor with length less than $q_1$ is clearly $0^{a_1}$. Therefore by
    \autoref{thm:power_characterization} it is sufficient to analyze the largest fractional power of a (primitive)
    factor of length $q_k$ for $k \geq 1$. By \autoref{thm:power_characterization} the index $a_{k+1}+2$ of the first
    $q_{k-1} - 1$ conjugates of $\mirror{s}_k$ dominates the index of the rest of the factors of length $q_k$. The
    fractional part of the fractional index of a factor $w$ is determined by the shortest extension of $w$ to a right
    special factor. Note that from the proof of \autoref{thm:conjugate_interval_lengths} it is evident that
    $C^{q_{k-1} - 2}(\mirror{s}_k) = s_k$. Thus among the first $q_{k-1} - 1$ conjugates of $\mirror{s}_k$, the factor
    $s_k$ has longest extension to a right special factor, and the length of the extension is $q_{k-1} - 2$. Thus the
    fractional index of $s_k$ is $a_{k+1} + 2 + (q_{k-1} - 2)/q_k$. The claim follows.
  \end{proof}

  In particular, this theorem says that that a Sturmian word has bounded fractional index if and only if the partial quotients of
  its slope are bounded. This is a result of Mignosi \cite{1989:infinite_words_with_linear_subword_complexity}. An
  alternative proof was given by Berstel \cite{1999:on_the_index_of_sturmian_words}.

  \section{Acknowledgments}
  The author was supported by University of Turku Graduate School UTUGS Matti programme.

  We thank Markus Whiteland, Tero Harju, and Luca Zamboni for giving fruitful feedback on the draft version of this paper.

  \printbibliography
\end{document}